\documentclass[submission,copyright,creativecommons]{eptcs}
\providecommand{\event}{GANDALF 2019} 
\usepackage{underscore}           

\usepackage{amsmath}
\usepackage[utf8]{inputenc}
\usepackage[T1]{fontenc}
\usepackage{amsfonts}
\usepackage{caption}
\usepackage{amssymb}
\usepackage{float}
\usepackage{capt-of}
\usepackage{url}
\usepackage{float}
 \usepackage[all]{xy}
 \usepackage{enumitem}
 \usepackage{dsfont}
 \usepackage{stmaryrd}
\usepackage{pxfonts}
\usepackage{bbm}
\usepackage{tikz}
\usepackage[fancy]{tikz-inet}
\usetikzlibrary{automata}
\usetikzlibrary{arrows}
\usetikzlibrary{shapes}
\usetikzlibrary{decorations.pathmorphing}
\usetikzlibrary{fit}
\usepackage{ifthen}
\usepackage{stmaryrd} 

\tikzstyle{every picture}=[
  >=stealth', 
shorten >=1pt, 
node distance=1.44cm,
auto,
bend angle=45,
initial text=,
  every state/.style={inner sep=0.75mm, minimum size=1mm},
font=\small,
]

\def\ostar{\textcircled{$\star$}}

\def\sc{\mathrm{sc}}
\def\mon{\mathrm{Mon}}

\newcommand{\IntEnt}[1]{\llbracket #1\rrbracket}

\newtheorem{definition}{Definition}
\newtheorem{lemma}{Lemma}
\newtheorem{example}{Example}
\newtheorem{proposition}{Proposition}
\newtheorem{theorem}{Theorem}

\newtheorem{remark}{Remark}
\newtheorem{property}{Property}
\newtheorem{corollary}{Corollary}

\newenvironment{proof}{\textbf{Proof:}}{\hfill$\square$\newline}

\title{Algebraic and Combinatorial Tools for State Complexity : Application to the Star-Xor Problem}
\author{Pascal Caron
  \and Edwin Hamel-de-le-court
    \and Jean-Gabriel Luque
\institute{LITIS, Université de Rouen,\\ Avenue de l'Université,\\ 76801 Saint-\'Etienne du Rouvray Cedex,\\ France}
\email{\{Pascal.Caron, Edwin.Hamel-de-le-court, Jean-Gabriel.Luque\}@univ-rouen.fr}
}
\def\titlerunning{Algebraic Tools for State Complexity}
\def\authorrunning{P. Caron, E. Hamel-de le Court \& J.-G. Luque}
\def\copyrightholders{Caron, Hamel-de le Court \& Luque}

\begin{document}
\maketitle
\begin{abstract}
We investigate the state complexity of the star of symmetrical differences using modifiers and monsters.
A monster is an automaton in which every function from states to states is represented by at least one letter. A modifier is a set of functions allowing one to transform a set of automata into one automaton.
These recent theoretical concepts allow one to find easily the desired  state complexity. We then exhibit a witness with a constant size alphabet. 
\end{abstract}

\section{Introduction}

The state complexity of a rational language is the size of its minimal automaton and the state complexity of a rational operation is the maximal one of those languages obtained by applying this operation onto languages of fixed state complexities. 

The classical  approach is to compute an upper bound  and to provide a witness, that is a specific example reaching the bound which is then the desired state complexity.

Since the $70s$, the state complexity of
numerous unary and binary operations has been computed. See, for example, \cite{Dom02,GMRY17,JJS05,Jir05,JO08,Yu01a} for a survey of the subject. More recently, the state complexity of combinations of operations has also been studied. In most  cases the result is not simply the mathematical composition of the individual complexities and studies lead to interesting situations. Examples can be found in \cite{CGKY11,GSY08,JO11,SSY07}.

In some cases, the classical method has to be enhanced by two independent approaches. The first one consists in describing   states  by  combinatorial objects. Thus the upper bound is computed using combinatorial tools. For instance, in \cite{CLMP15}, the states are represented by tableaux representing boolean matrices and an upper bound for the catenation of symmetrical difference is given. These combinatorial objects will be used to compute an upper bound for the Kleene star of symmetrical difference.
The second one is an algebraic method consisting in 
 %
%
 building a witness for a certain class of rational operations by searching in a set of automata with as many  transition functions as possible. This method has the advantage of being applied to a large class of operations, but has the drawback of giving witnesses that have alphabets of non-constant size. Witnesses with small alphabets are indeed favoured in this area of research when they can be found, as evidenced by several studies (\cite{CLP16,CLP17}).
This approach has been described independently by Caron \textit{et al.} in \cite{CHLP18}  as the monster approach and by Davies in \cite{Dav18} as the OLPA (One Letter Per Action) approach but was implicitly present in older papers like \cite{BJLRS16,DO09}.

In this paper, we illustrate these approaches to find the state complexity of the star of symmetrical difference.  Furthermore, we improve the witness found by drastically reducing the size of its alphabet to a constant size.

The paper is organized as follows. Section \ref{sect-prel}  gives definitions and notations about automata and combinatorics. In Section  \ref{sec-mod}, we recall the monster approach : we define  \emph{modifiers}, \emph{monsters}, and  give  some properties of these structures related to state complexity.  In Section \ref{sec-sc},  the state complexity of star of symmetrical difference is computed. Hence, in  Section \ref{sec-borne}, we find witnesses for this operation with an alphabet size of $17$.

\section{Preliminaries}\label{sect-prel}
\subsection{Operations over sets}
The \emph{cardinality} of a finite set $E$ is denoted by $\#E$,  the \emph{set of subsets} of  $E$ is denoted by $2^E$ and the \emph{set of mappings} of $E$ into itself is denoted by $E^E$. The \emph{symmetric difference} of two sets $E_1$ and $E_2$ is denoted by $\oplus$ and defined by $E_1\oplus E_2=(E_1\cup E_2)\setminus (E_1\cap E_2)$. For any positive integer $n$, let us denote $\{0,\ldots, n-1\}$ by $\llbracket n\rrbracket$. $\mathds{1}$ denotes the identity mapping, the set of which depends on context.
\subsection{Languages and automata}
Let $\Sigma$ denote a finite alphabet. A \emph{word} $w$ over  $\Sigma$ is a finite sequence of symbols of $\Sigma$. The \emph{length} of  $w$, denoted by $|w|$, is the number of occurrences of symbols of $\Sigma$ in $w$. For  $a\in \Sigma$, we denote by $|w|_a$ the number of occurrences of  $a$ in $w$.  The set of all finite words over $\Sigma$ is denoted by $\Sigma ^*$. A \emph{language} is a subset of $\Sigma^*$.

A \emph{complete and deterministic finite automaton} (DFA) is a $5$-tuple $A=(\Sigma,Q,i,F,\delta)$ where $\Sigma$ is the input alphabet, $Q$ is a finite set of states, $i\in Q$ is the initial state, $F\subseteq Q$ is the set of final states and $\delta$ is the transition function from  $Q\times \Sigma$ to $Q$ extended in a natural way from $Q\times \Sigma^*$ to $Q$. The cardinality of $A$ is the cardinality of its set of states, \emph{i.e.} $\#A=\#Q$. We will often use $\IntEnt{n}$ for some $n\in\mathbb{N}$ as the set of states for DFAs.


Let $A=(\Sigma,Q,i,F,\delta)$ be a DFA. A word $w\in \Sigma ^*$ is \emph{recognized} by the DFA $A$ if $\delta(i,w)\in F$. The \emph{language recognized} by a DFA $A$ is the set $\mathrm L(A)$ of words recognized by $A$. Two DFAs are said to be \emph{equivalent} if they recognize the same language.

For any word $w$, we denote by $\delta^w$ the function $q\rightarrow\delta(q,w)$. Two states $q_1,q_2$ of $D$ are \emph{equivalent} if for any word $w$ of $\Sigma^*$, $\delta(q_1, w)\in F$ if and only if $\delta(q_2, w)\in F$. This equivalence relation is called the \emph{Nerode equivalence} and is denoted by $q_1\sim_{Ner} q_2$. If two states are not equivalent, then they are called \emph{distinguishable}.

A state $q$ is \emph{accessible} in a DFA  if there exists a word $w\in \Sigma ^*$ such that $q=\delta(i,w)$. A DFA is  \emph{minimal} if there does not exist any equivalent  DFA  with less states and it is well known that for any DFA, there exists a unique minimal equivalent one (\cite{HU79}). Such a minimal DFA  can be  obtained from $D$ by computing $\widehat A_{/\sim}=(\Sigma,Q/\sim,[i],F/\sim,\delta_{\sim})$ where $\widehat A$ is the accessible part of $A$, and where, for any $q\in Q$, $[q]$ is the $\sim$-class of the state $q$ and satisfies the property  $\delta_{\sim}([q],a)=[\delta(q,a)]$, for any $a\in \Sigma$. The number of its states is denoted by $\#_{Min}(A)$. In a minimal DFA, any two distinct states are pairwise distinguishable.

Let $L$ be a regular language defined over an alphabet $\Sigma$. We denote by $L^*$ $\{w=u_1\cdots u_n\mid  u_i\in L \land n\in\mathbb{N}\}$.

The syntactic semigroup of $L$ is the semigroup generated by the transition functions of all letters of the minimal DFA of $L$.
\subsection{State complexity}
A \emph{unary regular operation} is a function from regular languages into regular languages of $\Sigma$. A \emph{$k$-ary regular operation} over the alphabet $\Sigma$ is a function from the set of $k$-tuples of regular languages of $\Sigma$ into regular languages of $\Sigma$.\\
The state complexity of a regular language $L$ denoted by $\sc(L)$ is the number of states of its minimal DFA. This notion extends to regular operations: the state complexity of a unary regular operation $\otimes$ is the function $\sc_{\otimes}$ such that, for all $n\in\mathbb{N\setminus\{0\}}$, $\sc_{\otimes}(n)$ is the maximum of all the state complexities of $\otimes(L)$ when $L$ is of state complexity $n$, \emph{i.e.} $\sc_{\otimes}(n)=\max\{\mathrm{sc}(\otimes(L)) | \sc(L) = n\}$.

This can be generalized, and the state complexity of a $k$-ary operation $\otimes$ is the $k$-ary function $\sc_\otimes$ such that, for all $(n_1,\ldots,n_k)\in (\mathbb N^*)^k$,
\begin{equation}
  \sc_\otimes(n_1,\ldots,n_k)=\max\{\sc(\otimes(L_1,\ldots,L_k))\mid\text{ for all }i\in\{1,\ldots,k\}, \sc(L_i)=n_i\}.
  \end{equation}
Then, a witness for $\otimes$ is a a way to assign to each $(n_1,\ldots,n_k)$, assumed sufficiently big, a k-tuple of languages $(L_1,\ldots,L_k)$ with $\mathrm{sc}(L_i)=n_i$, for all $i\in\{1,\ldots,k\}$, satisfying $\sc_\otimes(n_1,\ldots,n_k)=\sc(\otimes(L_1,\ldots,L_k))$.
\subsection{Morphisms}
Let $\Sigma$ and $\Gamma$ be two alphabets. A morphism is a function $\phi$ from $\Sigma^*$ to $\Gamma^*$ such that, for all $w,v\in\Sigma^*$, $\phi(wv)=\phi(w)\phi(v)$. Notice that $\phi$ is completely defined by its value on letters.

Let $L$ be a regular language over alphabet $\Sigma$ recognized by the DFA $A=(\Sigma,Q,i,F,\delta)$ and let $\phi$ be a morphism from $\Gamma^*$ to $\Sigma^*$. Then, $\phi^{-1}(L)$ is the regular language recognized by the DFA $B=(\Gamma,Q,i,F,\delta')$ where, for all $a\in\Gamma$ and $q\in Q$, $\delta'(q,a)=\delta(q,\phi(a))$. Therefore, note that we have
\begin{property}\label{prop-scmorph}
  Let $L$ be a regular language and $\phi$ be a morphism. We have $\mathrm{sc}(\phi^{-1}(L))\leq\mathrm{sc}(L)$.
\end{property}
We say that a morphism $\phi$ is \emph{$1$-uniform} if the image by $\phi$ of any letter is a letter. In other words, a $1$-uniform morphism is a (not necessarily injective) renaming of the letters and the only complexity of the mapping stems from mapping $a$ and $b$ to the same image, i.e., $\phi(a) = \phi(b)$.
\section{Monsters and state complexity}\label{sec-mod}
In \cite{Brz13}, Brzozowski gives a series of properties that would make a language $L_n$ of state complexity $n$ sufficiently complex to be a good candidate for constructing witnesses for numerous classical rational operations. One of these properties is that the size of the syntactic semigroup is $n^n$, which means that each transformation of the minimal DFA of $L_n$ can be associated to a transformation by some non-empty word. This upper bound is reached when the set of transition functions of the DFA is exactly the set of transformations from state to state. We thus consider the set of transformations of $\IntEnt{n}$ as an alphabet where each letter is simply named by the transition function it defines. This leads to the following definition :
\begin{definition}
A $1$-monster  is an automaton
$\mon_{n}^{F}=(\Sigma,\IntEnt{n},0, F,\delta)$ defined by
\begin{itemize}
\item   the  alphabet $\Sigma=\IntEnt{n}^{\IntEnt{n}}$, 
\item the set of states  $\IntEnt{n}$,
\item the initial state $0$,
\item the set of final states $F$,
\item the transition function $\delta$  defined for any $a\in \Sigma$ by $\delta(q,a)=a(q)$.
\end{itemize}
The language recognized by a $1$-monster DFA is called a \emph{$1$-monster language}.
\end{definition}
\begin{example}\label{ex-mon}
  The $1$-monster $\mon_2^{\{1\}}$ is
  
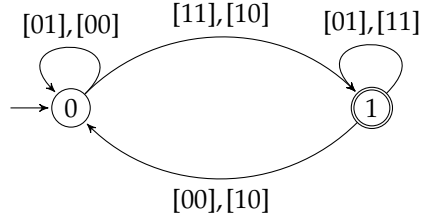
\begin{figure}[H]
  \centering
  \begin{tikzpicture}[node distance=2cm]
    \node[state,initial](p0){$0$};
    \node[state,accepting](p1) at (4,0) {$1$};
    \path[->]
    (p0)edge[loop ] node [swap]{$[01],[00]$} (p0)
    (p0)edge[bend left] node {$[11],[10]$} (p1)
    (p1)edge[loop ] node [swap]{$[01],[11]$} (p1)
    (p1)edge[bend left] node{$[00],[10]$}(p0);
  \end{tikzpicture}
  \end{figure}
  where, for all $i,j\in\{0,1\}$, the label $[ij]$ denotes the transformation sending $0$ to $i$ and $1$ to $j$, which is also a letter in the DFA above.
\end{example}

Let us notice that some families of $1$-monster languages are witnesses for the Star and Reverse operations (\cite{CHLP18}). The following claim is easy to prove and captures a universality-like property of $1$-monster languages:
\begin{property}\label{prop-res}
  Let $L$ be any regular language recognized by a DFA $A=(\Sigma,\IntEnt{n},0, F,\delta)$. The language $L$ is the preimage of $\mathrm L(\mon_n^F)$ by the $1$-uniform morphism $\phi$ such that, for all $a\in\Sigma$, $\phi(a)=\delta^a$, i.e.
  \begin{equation}
    L=\phi^{-1}(\mathrm{L}(\mon_n^F)).
  \end{equation}
\end{property}

This is an important and handy property that we should keep in mind. We call it the \emph{restriction-renaming} property.

We can wonder whether we can extend the notions above to provide witnesses for $k$-ary operators. In the unary case, the alphabet of a monster is the set of all possible transformations we can apply on the states. In the same mindset, a $k$-monster DFA is a $k$-tuple of DFAs, and its construction must involve the set of $k$-tuples of transformations as an alphabet. Indeed, the alphabet of a $k$-ary monster has to encode all the transformations acting on each set of states independently one from the others. This leads to the following definition :
 \begin{definition}\label{def-mon}
   A $k$-monster  is a $k$-tuple of automata $\mon_{n_1,\ldots, n_k}^{F_1,\ldots, F_k}=(\mathds{M}_1,\ldots, \mathds{M}_k)$ where\\
   $\mathds{M}_j=(\Sigma,\IntEnt{n_j},0, F_j,\delta_j)$ for $j\in \{1,k\}$ is defined by
\begin{itemize}
\item the common alphabet  $\Sigma=\IntEnt{n_1}^{\IntEnt{n_1}}\times \ldots \times \IntEnt{n_k}^{\IntEnt{n_k}}$, 
\item the set of states  $\IntEnt{n_j}$,
\item the initial state  $0$,
\item the set of final states  $F_j$,
\item the transition function $\delta_j$  defined for any $(a_1,\ldots, a_k)\in \Sigma$ by $\delta_j(q,(a_1,\ldots, a_k))={a_j}(q)$.
\end{itemize}
A $k$-tuple of languages $(L_1,\ldots,L_k)$ is called a \emph{monster $k$-language} if there exists a $k$-monster \\
$(\mathds{M}_1,\ldots,\mathds{M}_k)$ such that $(L_1,\ldots,L_k)=(\mathrm L(\mathds{M}_1),\ldots,\mathrm L(\mathds{M}_k))$.
\end{definition}
\begin{remark}\label{remark-min}
  When $F_j$ is different from $\emptyset$ and $Q_j$, $\mathds{M}_j$ is minimal.
\end{remark}
Definition \ref{def-mon} allows us to extend the restriction-renaming property in a way that is still easy to check.
\begin{property}
  Let $(L_1,\ldots,L_k)$ be a $k$-tuple of regular languages over the same alphabet $\Sigma$. We assume that each $L_j$ is recognized by the DFA $A_j=(\Sigma,\IntEnt{n_j},0,F_j,\delta_j)$. Let $\mon_{n_1,\ldots, n_k}^{F_1,\ldots, F_k}=(\mathds{M}_1,\ldots, \mathds{M}_k)$. For all $j\in\{1,\ldots,k\}$, the language $L_j$ is the preimage of $\mathrm{L}(\mathds{M}_j)$ by the $1$-uniform morphism $\phi$ such that, for all $a\in\Sigma$, $\phi(a)=(\delta_1^a,\ldots,\delta_k^a)$, i.e.
  \begin{equation}
    (L_1,\ldots,L_k)=(\phi^{-1}(\mathrm{L}(\mathds{M}_1)),\ldots,\phi^{-1}(\mathrm{L}(\mathds{M}_k))).
  \end{equation}
\end{property}
It has been shown that some families of $2$-monsters are witnesses for binary boolean operations and for the catenation operation \cite{CHLP18}. Many papers concerning state complexity actually use monsters as witnesses without naming them (e.g. \cite{BJLRS16}).
Therefore, a natural question arises : can we define a simple class of rational operations for which monsters are always witnesses ? This class should ideally encompass some classical regular operations, in particular the operations studied in the papers cited above. In the next section, we define objects that allow us to answer this question.
\section{Modifiers}
We first describe a class of regular operations for which monsters are always witnesses in the unary case. Once again, the restriction-renaming property comes in handy and gives us the intuition we need. We call \emph{$1$-uniform} any unary regular operation $\otimes$ that commutes with any $1$-uniform morphism, \emph{i.e.} for every regular language $L$ and every $1$-uniform morphism $\phi$, $\otimes(\phi^{-1}(L))=\phi^{-1}(\otimes(L))$. For example, it is proven in \cite{Dav18} that the Kleene star and the reverse are $1$-uniform. Suppose now that $\otimes$ is a unary $1$-uniform operation. Then, if $L$ is a regular language, $A=(\Sigma,\IntEnt{n},0,F,\delta)$ its minimal DFA, and $\phi$ the $1$-uniform morphism sending any letter of $\Sigma$ into its associated transition function in $A$, we have
\begin{equation}
  \otimes(L)=\otimes(\phi^{-1}(\mathrm{L}(\mon_{n}^F))=\phi^{-1}(\otimes(\mathrm{L}(\mon_{n}^F))).
\end{equation}
It follows that $\mathrm{sc}(\otimes(L))= \mathrm{sc}(\phi^{-1}(\otimes(\mathrm{L}(\mon_{n}^F))))\leq \mathrm{sc}(\otimes(\mathrm{L}(\mon_{n}^F)))$ by Property \ref{prop-scmorph}. In addition, Remark \ref{remark-min} implies that $\mathrm L(\mon_{n}^F)$ has the same state complexity as $L$. Therefore, we have
\begin{theorem}\label{th-mon}
  Any $1$-uniform operation admits a family of monster $1$-languages as a witness.
\end{theorem}

We now introduce the second central concept of our paper. In many cases, to compute state complexities, it is easier to describe regular operations as constructions on DFAs. We would therefore like to find a class of operations on DFAs, that are naturally associated to $1$-uniform operations. Such an operation on DFAs needs to have some constraints that are described in the following definitions.
\begin{definition}
  The \emph{state configuration} of a DFA $A=(\Sigma,Q,i,F,\delta)$ is the triplet $(Q,i,F)$.
\end{definition}
\begin{definition}
  A \emph{$1$-modifier} is a unary operation on DFA $\mathfrak m$ that produces a DFA such that :
  \begin{itemize}
  \item For any DFA $A$, the alphabet of $\mathfrak m(A)$ is the same as the alphabet of $A$.
  \item For any DFA $A$, the state configuration of $\mathfrak m(A)$ depends only on the state configuration of the DFA $A$.
  \item For any DFA $A$ over the alphabet $\Sigma$, for any letter $a\in\Sigma$, the transition function of $a$ in $\mathfrak m(A)$ depends only on the state configuration of the DFA $A$ and on the transition function of $a$ in $A$.
  \end{itemize}
\end{definition}
\begin{example}{The star modifier.}\label{ex-star}
  For all DFA $A=(\Sigma,Q,i,F,\delta)$, define $\mathfrak{Star}(A)=(\Sigma,2^{Q},\emptyset,\{E| E\cap F \neq \emptyset\}\cup\{\emptyset\},\delta_1)$, where $\delta_1$ is as follows : for all $a\in\Sigma$,
  \[\delta_1^a(\emptyset)=\left\{\begin{array}{ll}\{\delta^a(i)\}\text{ if }\delta^a(i)\notin F \\
  \{\delta^a(i),i\}\mbox{ otherwise }
  \end{array}\right.
  \mbox{ and, for all } E\neq\emptyset,\;
  \delta_1^a(E)=\left\{\begin{array}{ll}\delta^a(E)\text{ if }\delta^a(E)\cap F=\emptyset \\
  \delta^a(E)\cup\{i\}\mbox{ otherwise }
  \end{array}\right.\]
\end{example}
The modifier $\mathfrak{Star}$ describes the classical construction on DFA associated to the Star operation on languages, \emph{i.e.} for all DFA $A$, $\mathrm L(A)^*=\mathrm L(\mathfrak{Star}(A))$.
\begin{example}
  If we apply the modifier $\mathfrak{Star}$ to the modifier $\mon_2^{\{1\}}$ described in Example \ref{ex-mon}, we obtain the DFA drawn in Figure \ref{Stmon}.
  \begin{figure}[H]
    \centering
    \begin{tikzpicture}[node distance=2cm]
      \node[state,initial](p1) at (0,4) {$\emptyset$};
      \node[state,accepting](p2) at (6,0) {$\{1\}$};
      \node[state](p0) at (0,0) {$\{0\}$};
      \node[state,accepting](p3) at (6,4) {$\{0,1\}$};
      \path[->]
      (p1)edge node [swap]{$[01],[00]$} (p0)
      (p1)edge node {$[11],[10]$} (p3)
      (p0)edge[loop left] node [swap]{$[01],[00]$}
      (p0)edge[bend left=15] node {$[11],[10]$} (p3)
      (p2)edge node {$[10],[00]$} (p0)
      (p2)edge node [swap]{$[11],[01]$} (p3)
      (p3)edge[loop right] node [swap]{$[10],[01],[11]$}
      (p3)edge[bend left=15] node {$[00]$} (p0);
    \end{tikzpicture}
  \end{figure}
  \captionof{figure}{$\mathfrak{Star}(\mon_2^{\{1\}})$}\label{Stmon}
  From this, one deduces the action of the modifier $\mathfrak{Star}$ on any DFA with two states. For instance, applying $\mathfrak{Star}$ to DFA $C$ (Figure \ref{C}) gives the DFA described in Figure \ref{St(C)}.
  \newline
  \begin{minipage}{.5\textwidth}
    \centering
    \resizebox{.8\textwidth}{!}{
  \begin{tikzpicture}[node distance=2cm]
    \node[state,initial](p0){$0$};
    \node[state,accepting](p1) at (4,0) {$1$};
    \path[->]
    (p0)edge[loop ] node [swap]{$a$} (p0)
    (p0)edge[bend left] node {$b$} (p1)
    (p1)edge[loop ] node [swap]{$a,b$} (p1);
  \end{tikzpicture}
  }
    \captionof{figure}{The DFA $C$}\label{C}
  \end{minipage}
  \begin{minipage}{.5\textwidth}
    \centering
    \resizebox{\textwidth}{!}{
    \begin{tikzpicture}[node distance=2cm]
      \node[state,initial](p1) at (0,4) {$\emptyset$};
      \node[state,accepting](p2) at (6,0) {$\{1\}$};
      \node[state](p0) at (0,0) {$\{0\}$};
      \node[state,accepting](p3) at (6,4) {$\{0,1\}$};
      \path[->]
      (p1)edge node [swap]{$a$} (p0)
      (p1)edge node {$b$} (p3)
      (p0)edge[loop left] node [swap]{$a$}
      (p0)edge[bend left=15] node {$b$} (p3)
      (p2)edge node [swap]{$a,b$} (p3)
      (p3)edge[loop right] node [swap]{$a,b$} (p3);
    \end{tikzpicture}
    }
    \captionof{figure}{$\mathfrak{Star}(C)$}\label{St(C)}
  \end{minipage}
  Remark that to apply $\mathfrak{Star}$ to $C$, we just take the subautomaton of $\mathfrak{Star}(\mon_2^{\{1\}})$ with letters being exactly the transition functions of letters in $C$, and rename its letters by the letters of $C$ of which they are the transition functions. The transition labeled by $b$ in Figure \ref{C} is first assimilated to the transition $[11]$  in $\mon_2^{\{1\}}$ (see Example \ref{ex-mon}). Hence, the transition labeled by $b$ in $\mathfrak{Star}(C)$ is the same as the transition labeled by $[11]$ in $\mathfrak{Star}(\mon_2^{\{1\}})$ (Figure \ref{Stmon}).
\end{example}
\begin{theorem}\label{the-eq}
  A regular unary operation $\otimes$ is $1$-uniform if and only if there exists a $1$-modifier $\mathfrak m$ such that for any regular language $L$ and any DFA $A$ recognizing $L$, $\otimes(L)=\mathrm L(\mathfrak m(A))$.
\end{theorem}
\begin{proof}
  Let $\otimes$ be a $1$-uniform unary operation. We define a $1$-modifier $\mathfrak m$ as follows.
  For any DFA $A=(\Sigma,Q_A,i_A,F_A,\delta_A)$, we can rename its set of states so that $A$ becomes the DFA $D=(\Sigma,\IntEnt{n},0,F,\delta)$. Let us denote by $B=(\IntEnt{n}^{\IntEnt{n}},Q',i',F',\delta')$ the minimal DFA of $\otimes(\mathrm L(\mon_n^{F}))$.
  We set $\mathfrak m(A)=(\Sigma,Q',i',F',\tilde \delta')$, with $\tilde\delta'(q,a)=\delta'(q,\delta^a)$. 
  Notice that $\mathfrak m$ is indeed a $1$-modifier. First, $(Q',i',F')$ depends only on $(Q_A,i_A,F_A)$. Second, $\tilde \delta'^a$ depends only on $\delta^a$ and on $\delta'$, which in turn depend only on $(Q_A,i_A,F_A)$ and $\delta_A^a$.
  
  Furthermore, by construction, $\mathrm L(\mathfrak m(A))=\phi^{-1}(\mathrm L(B))$, where $\phi$ is the $1$-uniform morphism such that $\phi(a)=\delta_D^a$ for all $a\in\Sigma$. Therefore, we have $\mathrm L(\mathfrak m (A))=\phi^{-1}(\otimes(\mathrm L(\mon_n^F)))$. And, since $\otimes$ is $1$-uniform, we obtain $\mathrm L(\mathfrak m (A))=\otimes(\phi^{-1}(\mathrm L(\mon_n^F)))=\otimes(L)$.
  
Conversely, let $\otimes$ be a regular operation and let $\mathfrak m$ be a $1$-modifier such that for any regular language $L$ and any DFA $A$ recognizing $L$, $\otimes(L)=\mathrm L(\mathfrak m(A))$. We must prove that $\otimes$ is $1$-uniform. Let $\Gamma$ and $\Sigma$ be two alphabets. Consider a 1-uniform morphism $\phi$ from $\Gamma^*$ to $\Sigma^*$ and a language $L$ over $\Sigma$. Let $A=(\Sigma,Q,i,F,\delta)$ be any DFA recognizing $L$ and let $B=(\Gamma,Q,i,F,\tilde \delta)$ the DFA such that $\tilde \delta^a=\delta^{\phi(a)}$ for any letter $a\in\Gamma$. We have $\mathrm L(B)=\phi^{-1}(\mathrm L(A))$.

Let $\mathfrak m(A)=(\Sigma,Q_1,i_1,F_1,\delta_1)$ and $\mathfrak m(B)=(\Gamma,Q_2,i_2,F_2,\delta_2)$. Since the state configuration of $A$ is the same as the state configuration of $B$, we have $(Q_1,i_1,F_1)=(Q_2,i_2,F_2)$. Furthermore, because the transition function of any letter $a\in\Gamma$ in $B$ is also the same as the transition function of $\phi(a)$ in $A$, we have $\delta_2^a=\delta_1^{\phi(a)}$. Hence, $\mathrm L(\mathfrak m(B))=\phi^{-1}(\mathrm L(\mathfrak m(A)))$, which implies that $\otimes(\mathrm L(B))=\phi^{-1}(\otimes(A))$. Therefore, $\otimes(\phi^{-1}(\mathrm L(A)))=\phi^{-1}(\otimes(\mathrm L(A)))$, as expected.
\end{proof}
We extend the previous theorems by generalizing the definitions to $k$-ary operations.
\begin{definition}\label{def-uni}
  A $k$-ary regular operation $\otimes$ is called $1$-uniform if, for any $k$-tuple of rational languages $(L_1,\ldots,L_k)$, for any $1$-uniform morphism $\phi$, $\otimes(\phi^{-1}(L_1),\ldots,\phi^{-1}(L_k))=\phi^{-1}(\otimes(L_1,\ldots,L_k))$.
\end{definition}
Using the same arguments as in Theorem \ref{th-mon}, we find
\begin{theorem}\label{th-mon2}
  Any $k$-ary $1$-uniform operation admits a family of monster $k$-languages as a witness.
\end{theorem}
\begin{proof}
  Suppose now that $\otimes$ is a $k$-ary $1$-uniform operation. Then, if $(L_1,\ldots,L_k)$ is a $k$-tuple of regular languages over $\Sigma$, $(A_1,\ldots,A_k)$ the $k$-tuple of DFAs such that each $A_j=(\Sigma,Q_j,i_j,F_j,\delta_j)$ is the minimal DFA of $L_i$, and $\phi$ the $1$-uniform morphism such that, for all $a\in\Sigma$, $\phi(a)=(\delta_1^a,\ldots,\delta_k^a)$, and if $\mon_{n_1,\ldots,n_k}^{F_1,\ldots,F_k}=(\mathds{M}_1,\ldots,\mathds{M}_k)$, then $\otimes(L)=\otimes(\phi^{-1}(\mathrm{L}(\mathds{M}_1)),\ldots,\phi^{-1}(\mathrm{L}(\mathds{M}_k)))=\phi^{-1}(\otimes(\mathrm{L}(\mathds{M}_1),\ldots,\mathrm{L}(\mathds{M}_k)))$.
It follows that $\mathrm{sc}(\otimes(L))= \mathrm{sc}(\phi^{-1}(\otimes(\mathrm{L}(\mathds{M}_1),\ldots,\mathrm{L}(\mathds{M}_k))))\leq \mathrm{sc}(\otimes(\mathrm{L}(\mathds{M}_1),\ldots,\mathrm{L}(\mathds{M}_k)))$ by Property \ref{prop-scmorph}. In addition, each $\mathrm{L}(\mathds{M}_j)$ has the same state complexity as $L_j$.
\end{proof}
\begin{definition}\label{def-mod}
  A $k$-modifier is a $k$-ary operation on DFAs over the same alphabet that returns a DFA and such that :
  \begin{itemize}
  \item The alphabet of $\mathfrak m (A_1,...,A_k)$ is the same as the alphabet of each $A_j$.
  \item For any $k$-tuple of DFAs $(A_1,\ldots,A_k)$, the state configuration of $\mathfrak m (A_1,...,A_k)$ depends only on the state configurations of the DFAs $A_1,\ldots,A_k$.
  \item For any $k$-tuple of DFAs $(A_1,\ldots,A_k)$ where each DFA is over the alphabet $\Sigma$, for any letter $a\in\Sigma$, the transition function of $a$ in $\mathfrak m (A_1,\ldots,A_k)$ depends only on the state configurations of the DFAs $A_1,\ldots, A_k$ and on the transition functions of $a$ in each of the DFAs $A_1,...,A_k$.
  \end{itemize}
\end{definition}
\begin{example}\label{ex-xor}
  For all DFAs $A=(\Sigma,Q_1,i_1,F_1,\delta_1)$ and $B=(\Sigma,Q_2,i_2,F_2,\delta_2)$, define
  \[\mathfrak{Xor}(A,B)=(\Sigma,Q_1\times Q_2,(i_1,i_2),(F_1\times(Q_2\setminus F_2)\cup (Q_1\setminus F_1)\times F_2),(\delta_1,\delta_2))\]
\end{example}
The modifier $\mathfrak{Xor}$ describes the classical construction associated to the operation Xor on couples of languages, \emph{i.e} for all DFAs $A$ and $B$, $\mathrm L(A)\oplus \mathrm L(B)=\mathrm L(\mathfrak{Xor}(A,B))$.
\begin{theorem}\label{th-mod}
  A regular $k$-ary operation $\otimes$ is $1$-uniform if and only if there exists a $k$-modifier $\mathfrak m$ such that for any $k$-tuple of regular languages $(L_1,\ldots,L_k)$ and any $k$-tuple of DFAs $(A_1,\ldots,A_k)$ such that each $A_j$ recognizes $L_j$, we have $\otimes(L_1,\ldots,L_k)=\mathrm L(\mathfrak m(A_1,\ldots,A_k))$.
\end{theorem}
The proof of Theorem \ref{the-eq} can be easily adapted to $k$-ary operations.

The following proposition states the effects of composition on modifiers and $1$-uniform operations and directly stems from Definitions \ref{def-uni} and \ref{def-mod}.
\begin{proposition}\label{prop-comp}
  Let $\otimes_1$ be a $k_1$-ary $1$-uniform operation and $\otimes_2$ be a $k_2$-ary $1$-uniform operation. The $(k_1+k_2)$-ary operation defined by $\otimes(L_1,\ldots,L_{k_1+k_2})= \otimes_1(L_1,\ldots,L_{l},\otimes_2(L_{l+1},\ldots,L_{l+k_2}),L_{l+k_2},\ldots,L_{k_1+k_2})$ is $1$-uniform. Furthermore, if $\mathfrak m_1$ is a $k_1$-modifier associated with $\otimes_1$ and $\mathfrak m_2$ is a $k_2$-modifier associated with $\otimes_2$, the operation on $(k_1+k_2)$-tuples of DFAs defined by
  \begin{equation}
    \mathfrak m(A_1,\ldots,A_{k_1+k_2})=\mathfrak m_1(A_1,\ldots,A_{l},\mathfrak m_2(A_{l+1},\ldots,A_{l+k_2}),A_{l+k_2},\ldots,A_{k_1+k_2})
  \end{equation}
  is a modifier associated to $\otimes$.
\end{proposition}

\section{State complexity of the star of symmetrical difference}\label{sec-sc}

In this section, we compute the state complexity of the $2$-ary regular operation $L_1\ostar L_2=(L_1\oplus L_2)^*$. Examples \ref{ex-star} and \ref{ex-xor} together with Proposition \ref{prop-comp} show that $\ostar$ is $1$-uniform and that an associated modifier can be defined by $\mathfrak{StX}(A_1,A_2)=\mathfrak{Star}(\mathfrak{Xor}(A_1,A_2))$. To be more precise, if $A_1=(\Sigma, Q_1,i_1,F_1,\delta_1)$ and $A_2=(\Sigma,Q_2,i_2,F_2,\delta_2)$, then
\[\mathfrak{StX}(A_1,A_2)=(\Sigma,2^{Q_1\times Q_2},\emptyset,\{E\in 2^{Q_1 \times Q_2}\mid E\cap F\neq\emptyset \}\cup\{\emptyset\},\delta)\]
where $F=(F_1\times Q_2)\oplus(Q_1\times F_2)$ and, for all $a\in \Sigma$,
\[\delta^a(\emptyset)=\left\{\begin{array}{ll}\{(\delta_1^a(i_1),\delta_2^a(i_2))\}\text{ if }(\delta_1^a(i_1),\delta_2^a(i_2))\notin F \\
  \{(\delta_1^a(i_1),\delta_2^a(i_2)),(i_1,i_2)\}\mbox{ otherwise }
\end{array}\right.\]
\begin{center}
and, for all $E\neq\emptyset$, $\delta^a(E)=\left\{\begin{array}{ll}(\delta_1^a,\delta_2^a)(E)\text{ if }(\delta_1^a,\delta_2^a)(E)\cap F=\emptyset \\
(\delta_1^a,\delta_2^a)(E)\cup\{(i_1,i_2)\}\mbox{ otherwise. }
\end{array}\right.$
\end{center}

Theorem \ref{th-mon2} states that $\ostar$ admits a family of $2$-monsters as witness. For any positive integers $n_1,n_2$, let $(\mathds M_1,\mathds M_2)=\mon_{n_1,n_2}^{\{n_1-1\},\{0\}}$. We are going to show that, for all $(n_1,n_2)\in \mathbb {N^*}^2$, $(\mathrm L(\mathds M_1)),\mathrm L(\mathds M_2))$ is indeed a witness for $\ostar$. This allows us to compute its state complexity.
To be more precise, here is the outline of our proof. For any positive integers $n_1,n_2$, any $F_1,F_2\subseteq\IntEnt{n_1}\times\IntEnt{n_2}$, let us denote by $\mathrm{M}_{F_1,F_2}$ the DFA $\mathfrak{StX}(\mon_{n_1,n_2}^{F_1,F_2})$. We are going to minimize the DFA $\mathrm{M}_{\{n_{1}-1\},\{0\}}$ by first computing its accessible states, and then, restricting it to its accessible states, by computing its Nerode equivalence. We will therefore have computed the minimal DFA equivalent to $\mathrm{M}_{\{n_{1}-1\},\{0\}}$, and computing its size allows us to compute the state complexity of $\mathrm L(\mathrm{M}_{\{n_{1}-1\},\{0\}})$. We  then show that the state complexity of $\mathrm L(\mathrm{M}_{\{n_{1}-1\},\{0\}})$ is the greatest out of all the state complexities of $\mathrm L(\mathrm{M}_{F_1,F_2})$, with $(F_1,F_2)\subseteq\IntEnt{n_1}\times\IntEnt{n_2}$. Theorem \ref{th-mon2} allows us to conclude that the state complexity of $\mathrm L(\mathrm{M}_{\{n_{1}-1\},\{0\}})$ is indeed $sc_{\ostar}(n_1,n_2)$.
\subsection{Computing the accessible states of $\mathrm{M}_{\{n_{1}-1\},\{0\}}$}
In order to understand more easily the next proofs, we associate elements of $2^{\IntEnt{n_1}\times\IntEnt{n_2}}$ to boolean matrices of size $n_1\times n_2$. Such a matrice is called a tableau when crosses are put in place of $1$s, and $0$s are erased. We denote by the same letter the element of $2^{\IntEnt{n_1}\times\IntEnt{n_2}}$, the associated boolean matrix, and the associated tableau. If $T$ is an element of $2^{\IntEnt{n_1}\times\IntEnt{n_2}}$, we denote by $T_{x,y}$ the value of the boolean matrix $T$ at row $x$ and column $y$. Therefore, the three following assertions mean the same thing : a cross is at the coordinates $(x,y)$ in $T$, $T_{x,y}=1$, $(x,y)\in T$.

We say that a cross at coordinates $(x,y)$ in an element of $2^{\IntEnt{n_1}\times\IntEnt{n_2}}$ is in the final zone of $\mathrm{M_{F_1,F_2}}$ if  $(x,y) \in (F_1\times\llbracket n_2 \rrbracket) \oplus (\llbracket n_1 \rrbracket \times F_2)$. We remark that an element of $2^{\IntEnt{n_1}\times\IntEnt{n_2}}$ is final in $\mathrm{M_{F_1,F_2}}$ if and only if it has a cross in the final zone of $\mathrm{M_{F_1,F_2}}$.
We fix for the remainder of this section two positive integers $n_1$ and $n_2$.
\begin{lemma}\label{lemma-acc}
  The states of $\mathrm{M}_{\{n_{1}-1\},\{0\}}$ that are accessible are exactly the tableaux $T$ of size $n_1\times n_2$ such that, if $T$ has a cross in the final zone of $\mathrm{M}_{\{n_{1}-1\},\{0\}}$ , then $T$ has cross at $(0,0)$.
\end{lemma}
\begin{proof}
  It is easy to see by the definition of the transition function of $\mathfrak{StX}$ that every tableau $T$ with a cross in the final zone of $\mathrm{M}_{\{n_{1}-1\},\{0\}}$ and no cross at $(0,0)$ is not accessible.
  
  Let $\delta$ be the transition function of $\mathrm{M}_{\{n_{1}-1\},\{0\}}$. If $T$ is a tableau of size $n_1\times n_2$, let $\#_\mathrm{nf}T$ be the number of crosses of $T$ which are not in the final zone of $\mathrm{M}_{\{n_{1}-1\},\{0\}}$.
    Let us define an order $<$ on cross matrices as $T < T'$ if and only if $\#T < \#T'$ or $(\#T = \#T'$ and $\#_\mathrm{nf}T < \#_\mathrm{nf}T')$ .\\
    Let us prove every tableau $T$ of size $n_1\times n_2$ such that, if $T$ has a cross in the final $(\{n_{1}-1\},\{0\})$-zone, then $T$ has cross at $(0,0)$, is accessible by induction on non-empty cross matrices for the partial order $<$ (the empty cross matrix is the initial state of $\mathrm{M}_{\{n_{1}-1\},\{0\}}$, and so it is accessible).
    
    The only minimal cross matrix for non-empty matrices and the order $<$ is the cross matrix with only one cross at $(0,0)$. This is accessible from the initial state $\emptyset$ by reading the letter $(\mathds{1},\mathds{1})$. Let us notice that each letter is a couple of functions of $\IntEnt{n_1}^{\IntEnt{n_1}}\times \IntEnt{n_2}^{\IntEnt{n_2}}$.
    Now let us take a cross matrix $T'$, and find a cross matrix $T$ such that $T<T'$, and T' is accessible from $T$. We distinguish the cases :
    \begin{itemize}
    \item $T'$ has no cross in the final zone, except maybe at $(0,0)$.
      \begin{itemize}
      \item \emph{Case $T'_{n_1-1,0}=0$.} Let $(i,j)$ be the index of a cross of $T'$. 
        Let $(f,g)=((0,i),(0,j))$ where $(0,i)$ and $(0,j)$ denote transpositions, and let $T=(f,g)(T')$ where $(f,g)(T')=\{(f(i),g(j))\mid (i,j)\in T'\}$. As $(f,g)$ is a one-to-one transformation on $\llbracket n_{1}\rrbracket \times \llbracket n_{2}\rrbracket$, as $(f,g)(T')$ has a cross at $(0,0)$ and as $T'$ does not have any  crosses in the final zone, we have  $\delta^{(f,g)}(T)= (f,g)(T)=(f,g)(f,g)(T')=T' $.
        We also have $T<T'$ since $\#T=\#T'$ and $\#_\mathrm{nf}T<\#_\mathrm{nf}T'$.
      \item \emph{Case $T'_{n_1-1,0}=1$.} Let $(f,g)=((0,n_1-1),\mathds{1})$ and let $T=(f,g)(T')$. We have $\delta^{(f,g)}(T)=T'$, and $T<T'$ as $\#_\mathrm{nf}T < \#_\mathrm{nf}T'$.      
      \end{itemize}
    \item $T'$ has a cross in the final zone other than $(0,0)$.

      Let $(i,j)$ be  such a cross, and let $(f,g)=((0,i),(0,j))$. Let $T''$ be the cross matrix obtained from $T'$ by deleting the cross at $(0,0)$. Let $T=(f,g)(T'')$.
      As $(f,g)$ is still one-to-one on $\llbracket n_{1}\rrbracket \times \llbracket n_{2}\rrbracket$,  we have $T_{0,0}=((f,g)(T''))_{0,0}=T''_{i,j}=1$, and $((f,g)(T))=((f,g)((f,g)(T'')))=T''$. As $T''$ has a  cross in the final zone, we therefore have  $\delta^{(f,g)}(T)=T'$ and $T<T'$ as $\#T<\#T'$.
    \item The only cross of $T'$ which is in the final zone is $(0,0)$.
      \begin{itemize}
      \item \emph{Case $A$: there exists $j$ such that $T'_{0,j}=1$.}
        
        Let $(f,g)=(\left(n_1-1\atop 0\right),\mathds{1})$ and let $T_{i,j}=\left\{\begin{array}{ll}
        1 & \mbox{if } (i,j)=(0,0),\\
        T'_{0,j} & \mbox{if } i=n_1-1 \land j\neq 0,\\
        T'_{i,j} & \mbox{otherwise.}
        \end{array}\right.$
        
        It is easy to check that $\delta^{(f,g)}(T)=T'$, and $T<T'$ as $\#_\mathrm{nf}T < \#_\mathrm{nf}T'$.
              
        \begin{figure}
          \centerline{
            \begin{tikzpicture}[scale=0.4]
              \draw[step=1.0,black, thin] (0,0) grid (5,5);
              \draw[black,very thick] (0,1) -- (5,1);
              \draw[black,very thick] (1,0) -- (1,5);
              \node[scale=1.5] at (2.5,1.5) {$\times$};
              \node[scale=1.5] at (0.5,4.5) {$\times$};
              \node[scale=1.5] at (3.5,3.5) {$\times$};
              \node[scale=1.5,blue] at (2.5,4.5) {$\times$};
              \node[scale=1.5] at (1.5,3.5) {$\times$};
              \draw[->] (9,2.5) -- node[midway,above,scale=1.5] {\tiny $\left(\left(4\atop 0\right),\mathds{1}\right)$} (6,2.5) ;
              \draw[step=1.0,black, thin] (10,0) grid (15,5);
              \draw[black,very thick] (10,1) -- (15,1);
              \draw[black,very thick] (11,0) -- (11,5);
              \node[scale=1.5] at (12.5,1.5) {$\times$};
              \node[scale=1.5] at (10.5,4.5) {$\times$};
              \node[scale=1.5] at (13.5,3.5) {$\times$};
              \node[scale=1.5,red] at (12.5,0.5) {$\times$};
              \node[scale=1.5] at (11.5,3.5) {$\times$};
          \end{tikzpicture}}
          \caption{The two tableaux $T'$ and $T$ for case $A$.}
                \end{figure}
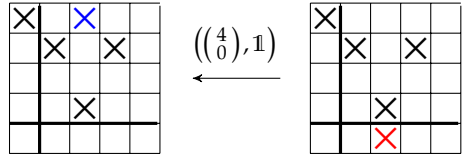
        
      \item \emph{ Case $\neg A$ and $T'_{n_1-1,0}=0$.} There exists $(i,j)\neq (n_1-1,0)$ such that $i\neq 0$ and $T'_{i,j}=1$. Let $(f,g)=((i,n_1-1),\mathds{1})$ and let $T=(f,g)(T')$. We have $\delta^{(f,g)}(T)=T'$, and $T<T'$ as $\#_\mathrm{nf}T < \#_\mathrm{nf}T'$.
        \begin{figure}
          \centerline{
            \begin{tikzpicture}[scale=0.4]
              \draw[step=1.0,black, thin] (0,0) grid (5,5);
              \draw[black,very thick] (0,1) -- (5,1);
              \draw[black,very thick] (1,0) -- (1,5);
              \node[scale=1.5] at (2.5,1.5) {$\times$};
              \node[scale=1.5] at (0.5,4.5) {$\times$};
              \node[scale=1.5,blue] at (3.5,3.5) {$\times$};
              \node[scale=1.5,blue] at (1.5,3.5) {$\times$};
              \draw[->] (9,2.5) -- node[midway,above,scale=1.5] {\tiny $\left(\left(1,4\right),\mathds{1}\right)$} (6,2.5) ;
              \draw[step=1.0,black, thin] (10,0) grid (15,5);
              \draw[black,very thick] (10,1) -- (15,1);
              \draw[black,very thick] (11,0) -- (11,5);
              \node[scale=1.5] at (12.5,1.5) {$\times$};
              \node[scale=1.5] at (10.5,4.5) {$\times$};
              \node[scale=1.5,red] at (13.5,0.5) {$\times$};
              \node[scale=1.5,red] at (11.5,0.5) {$\times$};
          \end{tikzpicture}}
          \caption{The two tableaux $T'$ and $T$ for case $\neg A$ and $T'_{n_1-1,0}=0$.}
        \end{figure}
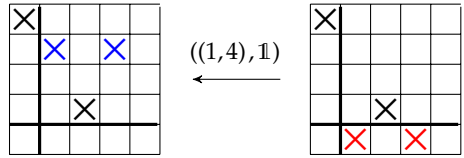
      \item \emph{Case $\neg A$ and $T'_{n_1-1,0}=1$.} Let $(f,g)=(\left(n_1-2,n_1-1\right),\mathds{1})$ and 
        let $T=(f,g)(T')$. We have $\delta^{(f,g)}(T)=T'$, and $T<T'$ as $\#_\mathrm{nf}T < \#_\mathrm{nf}T'$.      
      \end{itemize}
    \end{itemize}
    
\end{proof}
For all $(F_1,F_2)\subseteq\IntEnt{n_1}\times\IntEnt{n_2}$, let us now call $\mathrm{\widehat M}_{F_1,F_2}$ the DFA $\mathrm{M}_{F_1,F_2}$ restricted to states $T$ such that, if $T$ has a cross in the final zone of $\mathrm{M}_{F_1,F_2}$, then $T$ has cross at $(0,0)$. The following remark stems from the formula given for $\mathfrak{StX}$.
\begin{remark}\label{remark-acc}
  The accessible part of $\mathrm{M}_{F_1,F_2}$ is included in $\mathrm{\widehat M}_{F_1,F_2}$.
\end{remark}
\subsection{Computing the Nerode equivalence of $\mathrm{\widehat M}_{\{n_{1}-1\},\{0\}}$}
   
\begin{definition}
  A tableau $T$ in $2^{\IntEnt{n_1}\times \IntEnt{n_2}}$ is right-triangle free if $\forall x, x' \in \IntEnt{n_{1}}$ such that $x\neq x'$  and $\forall y,y' \in \IntEnt{n_{2}}$, such that $y\neq y'$, we have  $\#(\{(x,y),(x,y'),(x',y),(x',y')\}\cap T) \neq 3$.
\end{definition}

\begin{figure}
  \centerline{
    \begin{tikzpicture}[scale=0.5]
      \draw[step=1.0,black, thin] (0,0) grid (4,4);
      \draw[black] (0,2) -- (4,2);
      \draw[black] (3,0) -- (3,4);
      \node[scale=1.5] at (2.5,0.5) {$\times$};
      \node[scale=1.5] at (2.5,1.5) {$\times$};
      \node[scale=1.5] at (0.5,1.5) {$\times$};
      %
      %
  \end{tikzpicture}}
  \caption{A tableau with a right-triangle}
\end{figure}

\begin{definition}\label{def-fleche}
  If $T$ and $T'$ are distinct tableaux, we define the transformation on tableaux $\rightarrow$ as  $T \rightarrow T'$ if $T'=T \cup \{(i',j')\}$, and there exists $(i,j)$ such that  $\{(i,j),(i',j),(i,j')\}\subseteq T$. The equivalence relation $\overset{*}{\leftrightarrow}$ is defined as the  symmetric, reflexive and transitive closure of $\rightarrow$.
\end{definition}
   
For any tableau $T$, we define $\mathrm{Sat}(T)$ as the smallest tableau (relatively to inclusion) with no right-triangle containing $T$. The existence and the unicity of $\mathrm{Sat}(T)$ are easy to check. It is the representative of the equivalence class of $T$. Two tableaux $T$ and $T'$ are therefore equivalent if $\mathrm{Sat}(T)=\mathrm{Sat}(T')$.
\begin{lemma}\label{lemma-lines}
  The tableau $T$ in $2^{\IntEnt{n_1}\times \IntEnt{n_2}}$ is right-triangle free if and only if  for all $i,i' \in \llbracket n_{1}\rrbracket$, the lines $i$ and $i'$ are either the same  (for all $j\in \llbracket n_2 \rrbracket, T_{i,j}=T_{i',j}$), or disjoint  (for all $j\in \llbracket n_2 \rrbracket, T_{i,j}=0 \lor T_{i',j}=0$).
\end{lemma}
\begin{lemma}\label{lm-tool}
  Let $(F_1,F_2)\subseteq\IntEnt{n_1}\times\IntEnt{n_2}$, and let $T$ and $T'$ be any two states of $\mathrm{M_{F_1,F_2}}$ such that $T\rightarrow T'$. Then $T$ is final if and only if $T'$ is final.
\end{lemma}
Let us recall that the alphabet of $\mathrm{ M_{F_1,F_2}}$ is $\IntEnt{n_1}^{\IntEnt{n_1}}\times \IntEnt{n_2}^{\IntEnt{n_2}}$. If  $(f,g)$ is such a letter and  $T=\{(x_1,y_1),\ldots, (x_n,y_n)\}$ is a tableau, then define $(f,g)(T)$ as $\{(f(x_1),g(y_1)),\ldots, (f(x_n),g(y_n))\}$.
\begin{lemma}\label{lm-tool2}
  Let $(F_1,F_2)\subseteq\IntEnt{n_1}\times\IntEnt{n_2}$, and let $T$ and $T'$ be any two states of $\mathrm{\widehat M_{F_1,F_2}}$ such that $T\rightarrow T'$. Then, for any $a\in \IntEnt{n_1}^{\IntEnt{n_1}}\times \IntEnt{n_2}^{\IntEnt{n_2}} $, $\delta^a(T)\rightarrow \delta^a(T')$ or $\delta^a(T)= \delta^a(T')$.
\end{lemma}
\begin{proposition}\label{prop-undistinguishable}
  Let $(F_1,F_2)\subseteq\IntEnt{n_1}\times\IntEnt{n_2}$, and let $T,T'$ be two states of $\mathrm{ M_{F_1,F_2}}$. If $T \overset{*}{\leftrightarrow} T'$, then $T$ and $T'$ are not distinguishable.
\end{proposition}
\begin{proof}
  From Lemma \ref{lm-tool2}, it is easy to see by a simple induction that,  for any word $w$   if $T\rightarrow T'$ then $\delta^w(T)\rightarrow \delta^w(T')$ or $\delta^w(T)= \delta^w(T')$.
  From Lemma \ref{lm-tool}, if $T\rightarrow T'$, then $T\sim_{Ner} T'$ in the sense of the Nerode equivalence. Thus, as $\overset{*}{\leftrightarrow}$ is the symmetric and transitive closure of $\rightarrow$, $T\overset{*}{\leftrightarrow} T'$ implies $T\sim_{Ner} T'$.
\end{proof}
\begin{lemma}\label{lemma-dist}
  All states of $ (\mathrm{\widehat M}_{\{n_{1}-1\},\{0\}})_{/\overset{*}{\leftrightarrow}}$ are pairwise distinguishable.
\end{lemma}
\begin{proof}
  Let $\delta$ be the transition function of $ (\mathrm{\widehat M}_{\{n_{1}-1\},\{0\}})_{/\overset{*}{\leftrightarrow}}$.  Let $T$ and $T'$ be the representatives of two states of $(\mathrm{\widehat M}_{\{n_{1}-1\},\{0\}})_{/\overset{*}{\leftrightarrow}}$, such that $T\neq T'$. Let $(i,j)$ be such that $T_{i,j} \neq T'_{i,j}$. Suppose, for example that $T_{i,j}=1$. Take $\{i_{1},\dots,i_{\ell}\}=\{\alpha \mid  T'_{\alpha,j}=1\}$ and $\{j_{1},\dots,j_{p}\}=\{j\}\cup\{\beta \mid T'_{i_{1},\beta}=1\}$. We can see that :
  \begin{enumerate}
  \item By Lemma \ref{lemma-lines}, lines $i_{1},\dots, i_{\ell}$ are the same, as they all have a cross on the column $j$. Columns $\{j_{1},\dots,j_{p}\}$ are also the same, as they all have a cross on line $i_{1}$. It follows that, if $(i',j')\in \left(\{i_{1},\dots,i_{\ell}\}\times \left(\{0,\dots,n_2-1\}\setminus \{j_{1},\dots,j_{p}\} \right)\right)\cup \left(\left(\{0,\dots,n_1-1\}\setminus \{i_{1},\dots,i_{\ell}\} \right)\times \{j_{1},\dots,j_{p}\}\right)$, then $T'_{i',j'}=0$
  \item $j\in\{j_{1},\dots,j_{p}\}$ and $i\not\in\{i_{1},\dots,i_{\ell}\}$.
  \end{enumerate}

  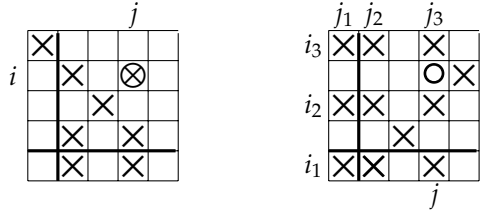
\begin{figure}
    \centerline{
                  \begin{tikzpicture}[scale=0.4]
                    \draw[step=1.0,black, thin] (0,0) grid (5,5);
                    \draw[black,very thick] (0,1) -- (5,1);
                    \draw[black,very thick] (1,0) -- (1,5);
                    \node[scale=1.5] at (0.5,4.5) {$\times$};
                    \node[scale=1.5] at (3.5,1.5) {$\times$};
                    \node[scale=1.5] at (1.5,1.5) {$\times$};
                    \node[scale=1.5,black] at (3.5,3.5) {$\otimes$};
                    \node[scale=1.5,black] at (1.5,3.5) {$\times$};
                    \node[scale=1.5,black] at (2.5,2.5) {$\times$};
                    \node at (3.5,5.5) {$j$};
                    \node at (-0.5,3.5) {$i$};
                    \draw[step=1.0,black, thin] (10,0) grid (15,5);
                    \draw[black,very thick] (10,1) -- (15,1);
                    \draw[black,very thick] (11,0) -- (11,5);
                    \node[scale=1.5] at (12.5,1.5) {$\times$};
                    \node[scale=2,black] at (13.5,3.5) {$\circ$};
                    \node[scale=1.5] at (11.5,4.5) {$\times$};
                    \node[scale=1.5] at (11.5,2.5) {$\times$};
                    \node[scale=1.5] at (11.5,0.5) {$\times$};
                    \node[scale=1.5] at (13.5,4.5) {$\times$};
                    \node[scale=1.5] at (13.5,2.5) {$\times$};
                    \node[scale=1.5] at (13.5,0.5) {$\times$};
                    \node[scale=1.5] at (14.5,3.5) {$\times$};
                    \node[scale=1.5] at (11.5,0.5) {$\times$};
                    \node[scale=1.5] at (10.5,4.5) {$\times$};
                    \node[scale=1.5] at (10.5,2.5) {$\times$};
                    \node[scale=1.5] at (10.5,0.5) {$\times$};
                    \node[scale=1.5] at (3.5,0.5) {$\times$};
                    \node[scale=1.5] at (1.5,0.5) {$\times$};
                    \node at (10.5,5.5) {$j_1$};
                    \node at (11.5,5.5) {$j_2$};
                    \node at (13.5,5.5) {$j_3$};
                    \node at (9.5,0.5) {$i_1$};
                    \node at (9.5,2.5) {$i_2$};
                    \node at (9.5,4.5) {$i_3$}; 
                    \node at (13.5,-0.5) {$j$};
                  \end{tikzpicture}}
                  \caption{An example of two tableaux $T$ and $T'$}
  \end{figure}

Let $f(i')=\left\{\begin{array}{ll}
n_1-1&\text{if }i'\in\{i_{1},\dots, i_{\ell}\},\\
0&\text{otherwise,}
	\end{array}\right.$
and
$g(j')=\left\{\begin{array}{ll}
0&\text{if }j'\in\{j_{1},\dots, j_{p}\}\\
n_2-1&\text{otherwise.}
\end{array}\right.$
 
If $(f(i'),g(j'))$ is in the final zone of $\mathrm{M}_{\{n_{1}-1\},\{0\}}$, then 
\begin{gather*}
  (i',j')\in \left(\{i_{1},\dots,i_{\ell}\}\times \left(\{0,\dots,n_2-1\}\setminus \{j_{1},\dots,j_{p}\} \right)\right) \cup \left(\left(\{0,\dots,n_1-1\}\setminus \{i_{1},\dots,i_{\ell}\} \right)\times \{j_{1},\dots,j_{p}\}\right),
\end{gather*} and so the first point above gives us $T'_{i',j'}=0$.

Therefore, $\delta^{(f,g)}(T')$ has only at most two crosses, one in $(n_1-1,0)$ and one in $(0,n_{2}-1)$, and it is not final. However, the second point above and the fact that $T_{i,j}=1$ gives us that $\delta^{(f,g)}(T)_{0,0}=1$, which means that $\delta^{(f,g)}(T)$ is final. Thus, $T$ and $T'$ are distinguishable. 
\end{proof}
Proposition \ref{prop-undistinguishable} and Lemma \ref{lemma-dist} give us that $ (\mathrm{\widehat M}_{\{n_{1}-1\},\{0\}})_{/\overset{*}{\leftrightarrow}}$ is the minimal DFA equivalent to the DFA $ \mathrm{\widehat M}_{\{n_{1}-1\},\{0\}}$. The following corollary stems from this assertion combined with Lemma \ref{lemma-acc}.
\begin{corollary}\label{cor-min}
  $ (\mathrm{\widehat M}_{\{n_{1}-1\},\{0\}})_{/\overset{*}{\leftrightarrow}}$ is the minimal DFA equivalent to $\mathrm{M}_{\{n_{1}-1\},\{0\}}$.
\end{corollary}
\subsection{Computing the state complexity of the language recognized by $\mathrm{M}_{\{n_{1}-1\},\{0\}}$}
The number of right-triangle free tableaux $T$ of size $\IntEnt{n_1}\times\IntEnt{n_2}$ such that, if $T$ has a cross in the final zone of $\mathrm{M}_{\{n_{1}-1\},\{0\}}$, then $T$ has cross at $(0,0)$ is exactly $2\alpha_{n_1-1,n_2-1}+\alpha'_{n_1,n_2}$ where $\alpha_{x,y}$ is the number of right-triangle free tableaux of size $x\times y$ and $\alpha'_{x,y}$ the number of right-triangle free tableaux of size $x\times y$ having a cross in $(0,0)$. Therefore,
\begin{lemma}\label{lemma-comp}
  The state complexity of $\mathrm L(\mathrm{M}_{\{n_{1}-1\},\{0\}})$ is $2\alpha_{n_1-1,n_2-1}+\alpha'_{n_1,n_2}$.
\end{lemma}
Closed formulas for $\alpha(x,y)$ and $\alpha'(x,y)$ are given in Corollary $20$ and Proposition $22$ of  \cite{CLMP15}.

In the next subsection, we prove that $(\{n_{1}-1\},\{0\})$ is a couple of final states that maximizes the size of the minimal DFA associated to any $\mathrm{M}_{F_1,F_2}$, with $(F_1,F_2)\subseteq\IntEnt{n_1}\times\IntEnt{n_2}$.
\subsection{Maximizing the state complexity of $\ostar$ applied to monster $2$-languages}
Let $\mathcal T$ be the set of right-triangle free tableaux of size $n_1\times n_2$. For all $(F_1,F_2)\subseteq\IntEnt{n_1}\times\IntEnt{n_2}$, let \[\mathcal T_{F_1,F_2}=\#(\mathrm{\widehat M}_{F_1,F_2})_{/\overset{*}{\leftrightarrow}}=\#\{T\in \mathcal T\mid T \text{ has a cross in the final } \text{zone implies } T_{0,0}=1\}.\]
We show that :
\begin{lemma}\label{lemma-max}
  For any $F_1\times F_2\subseteq \IntEnt{n_1}\times\IntEnt{n_2}$ such that $F_1,F_2\neq\emptyset$ and $F_1\neq\IntEnt{n_1}$,$F_2\neq\IntEnt{n_2}$, $\mathcal T_{F_1,F_2}\leq \mathcal T_{\{n_1-1\},\{0\}}$.
\end{lemma}
Therefore, by Remark \ref{remark-acc}, Proposition \ref{prop-undistinguishable}, and Corollary \ref{cor-min}, for any $F_1\times F_2\subseteq \IntEnt{n_1}\times\IntEnt{n_2}$ such that $F_1,F_2\neq\emptyset$ and $F_1\neq\IntEnt{n_1}$,$F_2\neq\IntEnt{n_2}$,
\[ \#_{\min}(\mathrm{M}_{F_1,F_2})\leq \#((\mathrm{\widehat M}_{F_1,F_2})_{/\overset{*}{\leftrightarrow}}) = \mathcal T_{F_1,F_2}\leq \mathcal T_{\{n_1-1\},\{0\}} = \#((\mathrm{\widehat M}_{\{n_1-1\},\{0\}})_{/\overset{*}{\leftrightarrow}}) = \#_{\min}(\mathrm{M}_{\{n_1-1\},\{0\}}).\]
The cases where $F_1=\emptyset$ or $F_2=\emptyset$ or $F_1=\IntEnt{n_1}$ or $F_2=\IntEnt{n_2}$ are easy and proven by :
\begin{lemma}\label{lemma-part}
  If $F_1=\emptyset$ or $F_2=\emptyset$ or $F_1=\IntEnt{n_1}$ or $F_2=\IntEnt{n_2}$, then $\#_{\min}(\mathrm{M}_{F_1,F_2})\leq \#_{\min}(\mathrm{M}_{\{n_1-1\},\{0\}})$.
\end{lemma}
Therefore, by Theorem \ref{th-mon2} and Lemma \ref{lemma-comp},
\begin{theorem}
  The state complexity of $\ostar$ is $2\alpha_{n_1-1,n_2-1}+\alpha'_{n_1,n_2}$, \emph{i.e.} for all $n_1,n_2\in \mathbb N^*$, $sc_{\ostar}(n_1,n_2)=2\alpha_{n_1-1,n_2-1}+\alpha'_{n_1,n_2}$.
\end{theorem}
\section{Witnesses with a bounded alphabet size}\label{sec-borne}
We now prove that there is a finite-bounded-alphabet witness. Let $n_1,n_2$ be two positive integers and let $(\mathds M_1,\mathds M_2)=\mon_{n_1,n_2}^{\{n_{1}-1\},\{0\}}$. Recall that the letters of $\mon_{n_1,n_2}^{\{n_{1}-1\},\{0\}}$ are couples of mappings and that $\mathds{1}$ is the identities bot in $\IntEnt{n_1}$ and in $\IntEnt{n_2}$. Let $B_1$ and $B_2$ be the DFAs obtained by restricting the letters of respectively  $\mathds M_1$ and $\mathds M_2$ to the alphabet
\[\begin{array}{ll}  \Sigma'=&\left\{((0,\ldots ,n_1-2),\mathds{1}), ((1,\dots, n_1-2),\mathds{1}),(\mathds{1},(1,\dots, n_2-2)),((1,\dots, n_1-1),\mathds{1}),\color{white}{\binom a b}\right.\\
          &(\mathds{1},(1,\dots,n_2-1)),
          ((0,n_1-1),\mathds{1}),(\mathds{1},(0,n_2-1)),((0,1),(0,1)),((0,1),\mathds{1}),(\mathds{1},(0,1)),\\
          &\left.((n_1-2,n_1-1),\mathds{1}), (\left(1\atop 0\right),\mathds{1}),
          (\mathds{1},\left(1\atop 0\right)),
          (\left(n_1-2\atop n_1-1\right),\mathds{1}),(\mathds{1},\left(n_2-2 \atop n_2-1\right)),(\left(n_1-1\atop 0\right),\mathds{1}),(\mathds{1},\left(n_2-1 \atop 0\right))\right\}.
\end{array}\]
Let $B=\mathfrak{StX}(B_1,B_2)$, and $\widehat B$ be the DFA obtained by restricting $B$ to states $T$ such that, if $T$ has a cross in the final zone of $\mathrm{M}_{\{n_1-1\},\{0\}}$, then $T$ has cross at $(0,0)$. The DFA $A=\widehat B_{/ \overset{*}{\leftrightarrow}}$ is obtained by restricting the letters of $\mathrm{\widehat M}_{\{n_{1}-1\},\{0\}}$ to the alphabet $\Sigma'$. We are going to show that $A$ is minimal.
        
Let us recall that all letters of $\Sigma'$ can be seen as a function acting on tableaux. Every word $w$ of $\Sigma'$ acts on a tableau $T$ by applying the composition of all letters of $w$ to $T$ : if $w=a_1\ldots a_n$, define $w(T)=a_n\circ\ldots\circ a_1(T)$. When it exists, we denote by $w^{-1}$ the inverse function of $a_n\circ\ldots\circ a_1$. Let $\delta$ be the transition function of $B$. We first notice that $w(T)$ is not necessarily equivalent to $T'=\delta^{w}(T)$ since $(0,0)$ is in $T'$ if $T'$ has a cross in the final zone. We denote by $w[i,j]$ the subword $a_i\cdots a_j$. By convention, if j<i, $w[i,j]=\varepsilon$. The proof of the following lemma is easy by induction.
\begin{lemma}\label{lemma-sat}
          Let $w$ be a word of $\Sigma'$, and $T$ be a state of $B$. If, for any integer $k<|w|$, we have              
$(w[1,k](T))_{0,0}=1$ or $w[1,k](T)$ has no cross in the final zone, then $\delta^{w}(T)=w(T)$.
\end{lemma}
        \begin{lemma}\label{lemma-acc-fin}
          All the states of $\widehat B$ are accessible.
        \end{lemma}
        \begin{proof}
        As the induction is the same as in Lemma \ref{lemma-acc}, we  only focus on cases of this previous lemma where the letters used are not in $\Sigma'$.
       \begin{itemize}
            \item If $T'$ has no cross in the final zone, according the previous remark, we have only to examine the case where $T'_{n_1-1,0}=0$. Let $(i,j)$ be the index of a cross of $T'$. Let $w=((\mathds{1},(0,1))((0,\ldots, n_1-2),\mathds{1})^i(\mathds{1},(1,\ldots, n_2-1))^{j-1}$ and let $T=w^{-1}(T')$. We have $T_{0,0}=1$ and for all $1<k<|w|$, for all $(i,j)\neq(n_1-1,0)$ such that $i=n_1-1$ and $j=0$, $(w[1, k](T))_{i,j}=(w[k,|w|-1]^{-1}(T'))_{i,j}=0$. Thus, by Lemma \ref{lemma-sat}, $\delta^{w}(T)=T'$. We also have $T<T'$ since $\#T=\#T'$ and $\#_\mathrm{nf}T<\#_\mathrm{nf}T'$.      
          \item $T'$ has a cross in the final zone other than $(0,0)$. Let $(i,j)$ be  such a cross. We distinguish two cases. 
          \begin{itemize}
          \item If $j=0$, we consider the word $w_1=((1,\ldots, n_1-2),\mathds{1})^i$. Let $T''$ be the tableau obtained from $w_1^{-1}(T')$ by deleting the cross at $(0,0)$ and let $w_2=((0,1),\mathds{1})$ and $T=w_2(T'')$. It is easy to see that $T'=w_1(\delta^{w_2}(T))$. By Lemma  \ref{lemma-sat}, we have $T'=w_1(\delta^{w_2}(T))=\delta^{w_2w_1}(T)$ in $\mathrm{\widehat M}_{\{n_{1}-1\},\{0\}}$ and $T<T'$ as $\#T<\#T'$.
          \item Otherwise define $w_1=((1,\ldots, n_1-1),\mathds{1})^{i-1}(\mathds{1},(1,\ldots, n_2-1))^{j-1}$. Let $T''$ be tableau obtained from $w_1^{-1}(T')$ by deleting the cross at $(0,0)$. It means that $T'=w_1(T'')\cup \{(0,0)\}$. Let $w_2=((0,1),(0,1))$ and $T=w_2(T'')$. Then we have $\delta^{w_2}(T)=T''\cup \{(0,0)\}$ or $\delta^{w_2}(T)=T''$.
            \begin{itemize}
            \item If $\delta^{w_2}(T)=T''\cup \{(0,0)\}$, then by Lemma \ref{lemma-sat}, as $w_1$ does not change the first line and the first column we have $\delta^{w_2w_1}(T)=\delta^{w_1}(\delta^{w_2}(T))=w_1(\delta^{w_2}(T))=w_1((T''\cup \{0,0)\})=w_1(T'')\cup \{(0,0)\}=T'$.\\
              \item If $\delta^{w_2}(T)=T''$, then we set $k=\min\{l<|w| \mid \delta^{(w_1[1,l]}(T'')_{0,0}=1\}$. For all integer $l<k$, the tableau $\delta^{w_1[1,l]}(T'')$ has no cross in the final zone, and we can apply Lemma \ref{lemma-sat}, and $w_2(w_1[1,l](T))=\delta^{w_2w_1}[1,l](T)$. Furthermore $\delta^{w_2w_1[1,k]}(T)=(w_1[1,k])(\delta^{w_2}(T))\cup \{(0,0)\}$ and as letters of $w_1[k+1,|w_1|]$ do not change the first line and the first column, we have $\delta^{w_2w_1}(T)=w_1[k+1,|w_1|](w_1[1,k](\delta^{w_2}(T))\cup \{(0,0)\})=w_1(\delta^{w_2}(T))\cup \{(0,0)\}=w_1(T'')\cup \{(0,0)\}=T'$.  Moreover, as  $\#T<\#T'$, we have $T<T'$.
            \end{itemize}
          \end{itemize}
        \item The only cross of $T'$ which is in the final zone is $(0,0)$. According to the first sentence of the proof, we have only to consider the case where there does not exist $j$ such that $T'_{0,j}=1$ and $T'_{n_1-1,0}=0$. It follows that there exists $(i,j)\neq (n_1-1,0)$ such that $i\neq 0$ and $T'_{i,j}=1$. Let $w=((1,\ldots, n_1-1),\mathds{1})^i$ and let $T=w^{-1}(T')$. By Lemma \ref{lemma-sat}, as for each proper prefix $w'$ of $w$, $(w'(T))_{0,0}=1$, we have $\delta^{w}(T)=T'$ in $B$, and $T<T'$ as $\#_\mathrm{nf}T < \#_\mathrm{nf}T'$.    
       \end{itemize}
        \end{proof}
        Similarly, the following lemma is obtained by simulating with letters in $\Sigma'$ the transition functions used in Lemma \ref{lemma-dist}.
        \begin{lemma}\label{lemma-dist-fin}
          All states of $A$ are pairwise distinguishable.
        \end{lemma}
        Lemma \ref{lemma-acc-fin} and \ref{lemma-dist-fin} imply that $A$ is minimal and that the following theorem holds.
        \begin{theorem}
          The couple $(B_1,B_2)$ is a witness for the operation $\ostar$.
        \end{theorem}
    \section{Conclusion}
        We have  given the state complexity of the star of symmetrical difference and have provided a witness with a constant alphabet size. We know that 
        the bounded size of the alphabet that we exhibit is not optimal, 
        but it simplifies the proof given. 
        Moreover, proving the optimality of a bound seems out of reach for now and would necessitate to introduce new  tools.
        
        One of our future works will be to generalize the method used here to a whole well-defined class of operations, in order to provide a witness with bounded alphabet size for all of them.
        {\bf Acknowlegedments }
        This work is partially supported by the projects MOUSTIC ( ERDF/GRR) and ARTIQ (ERDF/RIN)
        \bibliography{biblio}
\end{document}